\documentclass[a4paper,twocolumn,11pt,unpublished]{quantumarticle}
\pdfoutput=1

\usepackage[utf8]{inputenc}
\usepackage[english]{babel}
\usepackage[T1]{fontenc}
\usepackage{amsmath,amssymb,amsthm,mathtools}
\usepackage{microtype}
\usepackage{enumitem}
\usepackage{graphicx}
\usepackage[numbers,sort&compress]{natbib}\usepackage{hyperref}
\usepackage[nameinlink,noabbrev]{cleveref}
\usepackage{xcolor}
\usepackage{bm}

\newcommand{\Tr}{\mathrm{Tr}}
\newcommand{\KL}{\mathrm{KL}}

\newcommand{\cB}{\mathcal{B}}
\newcommand{\cF}{\mathcal{F}}
\newcommand{\id}{I}
\newcommand{\bbP}{\mathbb{P}}
\newcommand{\bbE}{\mathbb{E}}

\newcommand{\norm}[1]{\left\lVert #1 \right\rVert}

\theoremstyle{plain}
\newtheorem{theorem}{Theorem}
\newtheorem{lemma}{Lemma}
\newtheorem{corollary}{Corollary}

\theoremstyle{remark}

\hypersetup{
    colorlinks=true,
    linkcolor=blue,
    citecolor=blue,
    urlcolor=blue,
    pdftitle={Adaptive vs. Non-Adaptive Tomography under Pauli Basis Measurements: A Prefix/Tree Separation},
    pdfauthor={Alireza Goldar, Zhen Qin, Zhihui Zhu, Michael B. Wakin, and Zhe-Xuan Gong}
}

\title{An Exponential Advantage for Adaptive Tomography of Structured States under Pauli Basis Measurements}

\author{Alireza Goldar}
\affiliation{Department of Electrical Engineering, Colorado School of Mines, Golden, Colorado 80401, USA}

\author{Zhen Qin}
\affiliation{Michigan Institute for Computational Discovery and Engineering, Department of Electrical Engineering and Computer Science, and Department of Statistics, University of Michigan, Ann Arbor, Michigan 48109, USA}

\author{Zhihui Zhu}
\affiliation{Department of Computer Science and Engineering, The Ohio State University, Columbus, Ohio 43201, USA}

\author{Zhe-Xuan Gong}
\affiliation{Department of Physics, Colorado School of Mines, Golden, Colorado 80401, USA}

\author{Michael B. Wakin}
\affiliation{Department of Electrical Engineering, Colorado School of Mines, Golden, Colorado 80401, USA}

\begin{document}

\begin{abstract}
Broad claims about whether adaptivity helps in quantum state tomography can be misleading unless the state family, measurement architecture, and error metric are specified carefully. We study a restricted but physically important regime: single-copy quantum state tomography under local Pauli basis measurements, where the allowed measurement settings are tensor-product measurement operators built from local single-qubit Pauli operators, and performance is measured in trace distance with high probability in a minimax sense over a known structured family. We construct an explicit discrete prefix/tree family of states for which adaptive measurement selection achieves polynomial copy complexity, while every non-adaptive design requires exponentially many copies in the worst case. The adaptive upper bound comes from stagewise prefix recovery using hierarchical breadcrumb information revealed by partial prefix matches. The non-adaptive lower bound is based on a rare-prefix mechanism: every fixed design under-samples some deep prefix subset, and outside that subset the competing hypotheses induce identical one-shot laws, so only an exponentially small fraction of the measurement budget contributes to the KL divergence between the full data distributions. The result isolates a concrete regime in which adaptivity provably changes the sample-complexity scaling under the experimentally common local Pauli measurement architecture.
\end{abstract}

\maketitle

\section{Introduction}
\subsection{Overview}

The efficiency of quantum state tomography (QST) generally scales poorly with system size, motivating several broad lines of work to reduce its complexity, including structured quantum state tomography, protocols with stronger measurement resources, adaptive measurement protocols, and broader approaches to quantum-state learning and prediction \cite{Aaronson2007Learnability,AnshuArunachalam2024Survey,EisertHangleiterWalkRothMarkhamParekhChabaudKashefi2020,HuangKuengPreskill2020,ChenHuangLiLiuSellke2023,HuangBroughtonCotlerChenLiMohseniNevenBabbushKuengPreskillMcClean2022,ChenCotlerHuangLi2021,ChenGong2025,ChenGongYe2024,Yu2023,FlammiaODonnell2021,AaronsonChenHazanKaleNayak2018}. However, a central lesson from this literature is that the usefulness of adaptivity depends strongly on the exact measurement setting and learning objective, meaning that broad claims of ``adaptivity helps'' or  ``adaptivity does not help'' can often be misleading. In this paper, we identify a concrete measurement setting and quantum state state family in which adaptive measurement protocols provably changes the worst-case sample complexity scaling. Specifically, we demonstrate a rigorous exponential separation: for a structured family of multi-qubit states, an adaptive strategy achieves polynomial sample complexity, whereas every non-adaptive design requires exponentially many state copies. Crucially, this separation comes purely from sequential basis selection, without changing the underlying physical measurement architecture.

To be precise, we study individual-copy quantum state tomography under local Pauli basis measurements. The allowed measurement settings are tensor-product Pauli bases built from the three single-qubit observables $X$, $Y$, and $Z$, yielding a set $\cB = \{X,Y,Z\}^n$ of $3^n$ possible settings. One experimental shot consists of choosing a basis $b \in \cB$, measuring an unknown density matrix $\rho$ on $(\mathbb{C}^2)^{\otimes n}$, and observing an outcome $o \in \{0,1\}^n$. The resulting data are distributed according to Born's rule as $p(o \mid b,\rho) = \Tr\!\big(\Pi_{o}^{(b)} \rho\big)$, where $\Pi_{o}^{(b)}$ denotes the projector associated with outcome $o$ under measurement setting $b$. Our success criterion is trace-distance accuracy with high probability in a minimax sense over a known, fixed family $\cF$ of density matrices.

The separation is established over an explicit, structured ``prefix/tree'' family of states indexed by a hidden Pauli basis string. The key feature of this family is that information is distributed hierarchically across prefixes of the hidden string, rather than being concentrated entirely at the full length-$n$ level. As a result, partial agreement between a measurement basis and the hidden string is already informative: matching a correct prefix reveals a corresponding prefix-level signal, even if the remaining suffix does not match. This creates a breadcrumb structure that can be exploited sequentially by an adaptive procedure, in contrast to a non-hierarchical single-spike setting where mismatched measurements reveal nothing.

For this prefix/tree family, we prove two complementary results. First, we provide a constructive adaptive strategy that performs stagewise prefix recovery. At stage $k$, the procedure tests the three possible next Pauli symbols using a prefix statistic, yielding a polynomial copy complexity. Under our concrete parameterization, the adaptive procedure recovers the state up to trace distance $\varepsilon$ using $\widetilde{O}(n^3 \varepsilon^{-2})$ copies. Second, we prove that every non-adaptive design is exponentially costly. Any fixed non-adaptive allocation necessarily under-samples some long-prefix subset; by analyzing a hard two-point subproblem inside that subset, we show that non-adaptive strategies require an exponential $\Omega\big(3^n \varepsilon^{-2} \log(1/\eta)\big)$ number of copies in the worst case to succeed with probability $1-\eta$.

Placing this result in context, our theorem isolates a regime where adaptivity fundamentally changes the worst-case sample complexity within a specific structured family. The prefix/tree family is primarily a proof-oriented construction designed to isolate a genuine limitation of non-adaptive local Pauli tomography. However, it connects naturally to broader themes in structured tomography~\cite{JamesonQinGoldarWakinZhuGong2024,LovitzLowe2024,LangeKebricBuserSchollwockGrusdtBohrdt2023} and to commuting-Pauli and stabilizer-type constructions, which play important roles in quantum error correction, measurement-based quantum computation, and many-body physics~\cite{GuOlivieroLeone2024,Gottesman1997,RaussendorfBriegel2001,Kitaev2003}. 

\subsection{Related Work}

One important line of work studies whether quantum states can be learned efficiently, both in broad quantum-learning formulations and in worst-case tomography under single-copy measurements \cite{Aaronson2007Learnability,AnshuArunachalam2024Survey}. In this direction, Chen et al.~\cite{ChenHuangLiLiuSellke2023} showed that for general states, adaptivity does not improve the minimax sample complexity for trace-distance learning, though it can help for infidelity. By contrast, our work focuses on a structured family where adaptivity fundamentally changes the minimax trace-distance scaling, highlighting how strongly the utility of adaptivity depends on the underlying state family and measurement restrictions.

A parallel line of research investigates the sometimes exponential improvements possible when the learner is granted stronger physical resources, such as quantum memory or entangled measurements among two or more state copies~\cite{HuangBroughtonCotlerChenLiMohseniNevenBabbushKuengPreskillMcClean2022,ChenCotlerHuangLi2021,ChenGong2025,ChenGongYe2024}. While these works show substantial gains under richer sensing models, our separation arises purely from sequential basis selection, with single-copy measurement in a local Pauli basis that is practical for current quantum hardware~\cite{HuangKuengPreskill2020,Preskill2018NISQ}.

Within the Pauli-basis measurement framework, prior work has extensively explored observable estimation and identity testing, where the difficulty often stems from only a small subset of settings revealing the informative signal~\cite{Yu2023,ChenGongYe2024}. While philosophically similar, our objective is different: we aim for a minimax adaptive-versus-non-adaptive separation for state identification over a carefully designed structured family rather than observable estimation.

From an algorithmic perspective, substantial work highlights the practical benefits of active learning and sequential design for empirical efficiency in tomography~\cite{HuszarHoulsby2012,MahlerRozemaDarabiFerrieBlumeKohoutSteinberg2013,GranadeFerrieFlammia2017,LangeKebricBuserSchollwockGrusdtBohrdt2023}. Our paper complements this literature by providing a rigorous theorem-level separation under a specific measurement restriction, rather than proposing an empirical heuristic.

Our approach also connects to structured low-complexity tomography under local measurements, including tensor network-inspired and compressed sensing-based settings~\cite{CramerPlenioFlammiaSommaGrossBartlettLandonCardinalPoulinLiu2010,GrossLiuFlammiaBeckerEisert2010,FlammiaGrossLiuEisert2012,JamesonQinGoldarWakinZhuGong2024,LovitzLowe2024}, which show that exploiting structure fundamentally changes the learning problem. Although our structured family is introduced primarily as a discrete, proof-oriented construction, it can be represented in matrix product operator (MPO) form, providing a natural bridge to this literature and to the broader tensor-network formalism for mixed-state representations \cite{VerstraeteGarciaRipollCirac2004}.

Finally, some of the closest neighboring ideas involve prefix- or tree-related constructions. For example, Flammia and O'Donnell use a branch-and-prune prefix motif for learning Pauli channel error distributions~\cite{FlammiaODonnell2021}, and Aaronson et al.\ discuss prefix-dependent measurement selection in an online-learning framework~\cite{AaronsonChenHazanKaleNayak2018}. However, the role of the prefix here is distinct: our hierarchy is embedded directly into the state family itself through the observables $P_{b^\star}^{(k)}$ and coefficients $\beta_k$. This structural embedding uniquely drives both sides of our theorem, enabling stagewise adaptive recovery as well as the non-adaptive lower bound via the rare deep-prefix subset. To our knowledge, this exact combination of family design, measurement restriction, and rigorous exponential separation has not been established before.

\subsection{Contributions and Paper Organization}

This paper makes four contributions. It introduces a single-copy product-Pauli tomography problem over an explicit discrete prefix/tree family that is structurally different from spike-only Pauli constructions; it develops a constructive adaptive policy that exploits prefix-level information and achieves polynomial sample complexity; it proves a worst-case exponential lower bound for every non-adaptive Pauli design; and it places these results within the broader literature on adaptive tomography, Pauli-restricted learning, and structured quantum-state inference. The paper is organized accordingly. We first formalize the measurement model, the trace-distance success criterion, and the prefix/tree state family. We then present the adaptive upper bound via stagewise prefix recovery, followed by the non-adaptive lower bound via the rare-subset argument. We conclude by discussing how the separation fits into the surrounding literature and what directions it suggests for future work.

\section{Problem Setup and Main Results}
\label{sec:setup_main}

We now formalize the measurement model, the structured-family setting, and the precise claims proved in the remainder of the paper.

\subsection{Measurement model and protocol classes}
\label{subsec:measurement_model}

Let $n$ be the number of qubits and let $D := 2^n$ denote the Hilbert-space dimension. Throughout the paper, the physical measurement architecture is fixed to \emph{single-copy} tensor-product Pauli measurements, a local measurement setting that is especially natural for near-term and measurement-limited quantum platforms \cite{HuangKuengPreskill2020}. Thus the allowed measurement settings are
\[
\cB = \{X,Y,Z\}^{n},
\qquad
|\cB| = 3^n.
\]
A single shot proceeds as follows: one chooses a basis $b=(b_1,\dots,b_n)\in\cB$, measures one fresh copy of the unknown state $\rho$, and observes an outcome $o\in\{0,1\}^n$. The outcome law is given by Born's rule,
\[
p(o\mid b,\rho) = \Tr\!\big(\Pi_o^{(b)} \rho\big),
\]
where $\{\Pi_o^{(b)}\}_{o\in\{0,1\}^{n}}$ are the rank-one projectors associated with the product eigenbasis of $\bigotimes_{i=1}^n \sigma_{b_i}$. We let $M$ denote the total number of such single-shot measurements.

A measurement protocol using $M$ state copies is called \emph{adaptive} if the basis selected at a time $t \in \{2,\dots,M\}$ may depend on the past transcript
\[
(B_1,O_1),\dots,(B_{t-1},O_{t-1}).
\]
Here $B_i$ denotes the chosen measurement basis at round $i$, and $O_i$ denotes the corresponding observed outcome.

A measurement protocol is called \emph{non-adaptive} if the entire measurement schedule is fixed before any outcomes are observed. This includes deterministic schedules as well as randomized schedules, provided that the randomization is independent of the measurement outcomes. A non-adaptive protocol may therefore distribute its measurement budget arbitrarily over $\cB$, but it may not update that allocation online using the observed data.

The comparison studied in this paper is therefore \emph{not} between different measurement architectures. Both sides are restricted to the same single-copy product-Pauli measurement architecture; the only difference is whether the next basis can depend on previous outcomes.

\subsection{Structured-family formulation and success criterion}
\label{subsec:criterion}

Our goal is trace-distance accurate recovery with high probability over a known structured family. Specifically, for a family $\cF$ of density matrices, an $M$-copy protocol is said to be \emph{$(\varepsilon,\eta)$-accurate over $\cF$} if
\[
\sup_{\rho\in\cF}
\bbP_{\rho}\!\left[
\frac12 \norm{\widehat\rho - \rho}_{1} > \varepsilon
\right]
\le \eta.
\]
Thus the paper adopts a minimax, worst-case viewpoint over the family under study, in line with the modern sample-complexity perspective on rigorous quantum tomography guarantees \cite{HaahHarrowJiWuYu2017,ODonnellWright2016}.

A key point is that we are \emph{not} considering fully arbitrary unknown states. The family $\cF$ is fixed in advance and known to the learner; the unknown object is the discrete index selecting one member of that family. In the family used here, the coefficients are assumed known and only the hidden Pauli string index is unknown. Consequently, exact identification of that discrete index immediately yields exact recovery of the corresponding state within the family.

\subsection{The prefix/tree family}
\label{subsec:prefix_family_glance}

The family studied in this paper is indexed by a hidden basis string
\[
b^\star = (b^\star_1,\dots,b^\star_n) \in \cB.
\]
Based on $b^\star$, for each $k\in\{1,\dots,n\}$, we define an associated prefix Pauli operator
\[
P_{b^\star}^{(k)}
:=
\left(\bigotimes_{i=1}^{k} \sigma_{b^\star_i}\right)\otimes \id^{\otimes(n-k)}.
\]
For a fixed set of coefficients $\alpha,\beta_1,\dots,\beta_{n-1}\in\mathbb{R}$, we also define
\begin{equation}
\rho_{b^\star}
=
\frac{1}{D}
\left(
\id
+
\sum_{k=1}^{n-1}\beta_k P_{b^\star}^{(k)}
+
\alpha P_{b^\star}^{(n)}
\right),
\qquad
D=2^n.
\label{eq:prefix_family_main}
\end{equation}
Finally, we define the family of density matrices of interest as
\[
\cF_n(\alpha,\beta)
:=
\left\{
\rho_{b^\star} : b^\star\in\{X,Y,Z\}^n
\right\},
\]
where $\beta=(\beta_1,\dots,\beta_{n-1})$. Under a simple sufficient condition on $\alpha$ and $\beta$, each $\rho_{b^\star} \in \cF_n(\alpha,\beta)$ is a valid density matrix. The details of that condition, as well as the geometric interpretation of the family, are deferred to the next section. For now, the important point is that the family embeds information hierarchically across prefixes of the hidden string $b^\star$: a basis that matches the first $k$ symbols of $b^\star$ already reveals the prefix coefficient $\beta_k$, even if the remaining suffix does not match.

\subsection{Main results}
\label{subsec:main_results}

We now state the two theorem-level results that drive the paper: a constructive adaptive upper bound and a non-adaptive lower bound.

Throughout the prefix/tree analysis, the coefficients $\alpha,\beta_1,\dots,\beta_{n-1}$ are fixed and known to the learner; the only unknown quantity is the hidden Pauli string index $b^\star$.

We begin with the adaptive side, which shows that the prefix/tree structure can be exploited constructively to recover the hidden string with polynomial copy complexity.

\begin{theorem}[Adaptive copy complexity on the prefix/tree family]
\label{thm:adaptive_prefix}
Fix $n\ge 1$ and coefficients $\alpha,\beta_1,\dots,\beta_{n-1}$, assumed known to the learner, such that
\[
\alpha \neq 0, \quad
\beta_k \neq 0 \ \ \text{for all } k=1,\dots,n-1.
\]
Assume that the associated family $\cF_n(\alpha,\beta)$ is well defined as a family of density operators. Then for every confidence level $\eta\in(0,1)$, there exists an adaptive protocol under the single-copy product-Pauli measurement model such that, for every state $\rho_{b^\star}\in\cF_n(\alpha,\beta)$, the protocol identifies the hidden index $b^\star$ correctly with probability at least $1-\eta$ using at most
\[
M_{\mathrm{ad}}
\le
24 \log\!\Big(\frac{6n}{\eta}\Big)
\left(
\sum_{k=1}^{n-1}\beta_k^{-2} + \alpha^{-2}
\right)
\]
copies.

In particular, once the index is identified, the protocol can output $\widehat\rho=\rho_{b^\star}$ exactly. Hence the same protocol is trivially $(\varepsilon,\eta)$-accurate over $\cF_n(\alpha,\beta)$ for any target accuracy $\varepsilon>0$; see Corollary~\ref{cor:adaptive_prefix_tomography} for the corresponding tomography-form statement.
\end{theorem}

The next theorem gives the complementary non-adaptive lower bound: under the same single-copy product-Pauli measurement architecture, every non-adaptive design requires exponentially many copies in the worst case over the same family.

\begin{theorem}[Non-adaptive lower bound]
\label{thm:nonad_lb_full}
There exist universal constants $c_0,c_1>0$ such that the following holds. Let $\eta\in(0,1/4]$ and let $\cF_n(\alpha,\beta)$ be the prefix/tree family defined in \eqref{eq:prefix_family_main}. Assume that \(\alpha\neq 0,\) that $\cF_n(\alpha,\beta)$ is well defined as a family of density operators, and that there exists a constant $\delta\in(0,1)$ such that \(\sum_{k=1}^{n-1} |\beta_k| \le 1-\delta.\) Then any non-adaptive protocol under the same single-copy product-Pauli measurement model that is $(c_0 |\alpha|,\eta)$-accurate over $\cF_n(\alpha,\beta)$ must use at least
\[
M_{\mathrm{nonad}}
\ge
c_1 \,
\frac{3^{n-1}}{\mathrm{kl}(\alpha)}
\log\!\Big(\frac{1}{\eta}\Big)
\]
copies, where
\[
\mathrm{kl}(\alpha)
:=
\frac{1+\alpha}{2}\log(1+\alpha)
+
\frac{1-\alpha}{2}\log(1-\alpha).
\]
In particular, if $|\alpha|$ is sufficiently small, then \(\mathrm{kl}(\alpha)=\Theta(\alpha^2),\) and hence
\[
M_{\mathrm{nonad}}
=
\Omega\!\Big(\frac{3^n}{\alpha^2}\log(1/\eta)\Big).
\]

Equivalently, for every non-adaptive protocol using fewer than the above number of copies, there exists a member of the family $\cF_n(\alpha,\beta)$ such that \(\mathbb{P}_{\rho}\!\left(
\frac12\|\widehat\rho-\rho\|_1 > c_0|\alpha|
\right)\ge \eta.\)
\end{theorem}

The following corollary instantiates the two general bounds with a concrete coefficient choice, making the resulting polynomial-versus-exponential separation explicit in terms of $n$, $\varepsilon$, and $\eta$.

\begin{corollary}[Concrete polynomial-versus-exponential separation]
\label{cor:adaptive_nonadaptive_separation}
Fix a target accuracy scale $\varepsilon\in(0,1)$ and choose
\[
\alpha = \frac{\varepsilon}{4},
\qquad
\beta_k = \frac{\varepsilon}{4(n-1)},
\qquad
k=1,\dots,n-1.
\]
Then the family $\cF_n(\alpha,\beta)$ is physical and lies at trace-distance scale $\Theta(\varepsilon)$ from the maximally mixed state. Moreover, there exist universal constants $C,c>0$ such that:

\begin{enumerate}[label=(\roman*)]
\item there is an adaptive protocol with
\[
M_{\mathrm{ad}}
\le
C\, n^3 \varepsilon^{-2}\log\!\Big(\frac{n}{\eta}\Big)
=
\widetilde O(n^3 \varepsilon^{-2})
\]
copies that succeeds for every state in the family with probability at least $1-\eta$; while

\item every non-adaptive protocol that is $(c\varepsilon,\eta)$-accurate for every state in the same family must satisfy
\begin{equation*}
\begin{split}
M_{\mathrm{nonad}}
&\ge
c\, 3^{\,n-1}\varepsilon^{-2}
\log\!\Big(\frac{1}{\eta}\Big) \\
&=
\Omega\;\!\big(3^n \varepsilon^{-2}\log(1/\eta)\big).
\end{split}
\end{equation*}
\end{enumerate}
\end{corollary}

\paragraph{Informal summary.}
Theorems~\ref{thm:adaptive_prefix} and~\ref{thm:nonad_lb_full} together yield a sharp separation inside the same measurement architecture and over the same structured family. The adaptive protocol succeeds because the family reveals information hierarchically: partial prefix matches already produce detectable biases, allowing the adaptive protocol to recover the hidden string progressively, one prefix level at a time. By contrast, any non-adaptive design must commit its entire budget in advance, and the lower bound shows that every such fixed allocation leaves some deep prefix class exponentially under-sampled in the worst case. This is the source of the polynomial-versus-exponential gap.

\section{Prefix/Tree Family and Intuition}
\label{sec:prefix_family}

We now look more closely at the prefix/tree family and explain the ``breadcrumb'' mechanism that makes the adaptive upper bound possible while still permitting a worst-case non-adaptive lower bound.

\subsection{A simple physicality condition}
\label{subsec:prefix_family_definition}

The following elementary lemma records a convenient sufficient condition under which every member of $\cF_n(\alpha,\beta)$ is a valid density matrix.

\begin{lemma}[Sufficient physicality condition]
\label{lem:prefix_family_physicality}
If
\begin{equation}
|\alpha| + \sum_{k=1}^{n-1} |\beta_k| \le 1,
\label{eq:absum}
\end{equation}
then every matrix $\rho_{b^\star}$ defined by equation~\eqref{eq:prefix_family_main} is positive semidefinite and has unit trace.
\end{lemma}

\begin{proof}
Fix $b^\star$ and define
\[
\begin{aligned}
X_{b^\star}
&:=
\sum_{k=1}^{n-1}\beta_k P_{b^\star}^{(k)} + \alpha P_{b^\star}^{(n)},
\\
&\qquad \text{so that} \qquad
\rho_{b^\star}
=
\frac{1}{D}\bigl(\id + X_{b^\star}\bigr).
\end{aligned}
\]
Each operator $P_{b^\star}^{(k)}$ is a Hermitian Pauli string, and the family \(\left\{P_{b^\star}^{(k)}\right\}_{k=1}^n\) is pairwise commuting. Indeed, for any $j < k$,
\[
\begin{aligned}
P_{b^\star}^{(j)}
&=
\left(\bigotimes_{i=1}^{j}\sigma_{b^\star_i}\right)
\otimes \id^{\otimes(n-j)},
\\
P_{b^\star}^{(k)}
&=
\left(\bigotimes_{i=1}^{k}\sigma_{b^\star_i}\right)
\otimes \id^{\otimes(n-k)}.
\end{aligned}
\]
so the two strings differ only in the positions $j+1,\dots,k$, where one of them carries identity factors. Hence
$P_{b^\star}^{(j)} P_{b^\star}^{(k)}
= P_{b^\star}^{(k)} P_{b^\star}^{(j)}$. Since the prefix operators are Hermitian and commute pairwise, they admit a common orthonormal eigenbasis. Equivalently, there exists a single basis in which all matrices \(\{P_{b^\star}^{(k)}\}_{k=1}^n\) are diagonal simultaneously. Let $v$ be a common eigenvector. Since each $P_{b^\star}^{(k)}$ is Hermitian unitary, its eigenvalues are in $\{+1,-1\}$. Thus there exist signs $s_k(v) \in \{+1,-1\}$ for $k=1,\dots,n$ such that \(P_{b^\star}^{(k)} v = s_k(v)\, v.\) Applying $X_{b^\star}$ to $v$ gives
\[
X_{b^\star} v
=
\left(
\sum_{k=1}^{n-1}\beta_k s_k(v) + \alpha s_n(v)
\right) v.
\]
Hence every eigenvalue $\lambda_v$ of $X_{b^\star}$ has the form
\[
\lambda_v
=
\sum_{k=1}^{n-1}\beta_k s_k(v) + \alpha s_n(v),
\]
and therefore satisfies
\[
|\lambda_v|
\le
\sum_{k=1}^{n-1} |\beta_k| + |\alpha|.
\]
Under the assumption~\eqref{eq:absum}, we obtain \(\lambda_v \ge -1\) for every common eigenvector $v$. It follows that every eigenvalue of $\id + X_{b^\star}$ is nonnegative, because $( \id + X_{b^\star} ) v = (1+\lambda_v) v$ and $1+\lambda_v \ge 0$. Thus \(\id + X_{b^\star} \succeq 0,\) and consequently \(\rho_{b^\star} \succeq 0.\)

Finally, every non-identity Pauli string has trace zero, while \(\Tr(\id) = D.\) Since each $P_{b^\star}^{(k)}$ contains at least one non-identity Pauli factor, it follows that $\Tr\!\bigl(P_{b^\star}^{(k)}\bigr) = 0$ for $k=1,\dots,n$. Therefore
\[
\Tr(\rho_{b^\star})
=
\frac{1}{D}
\Tr\!\bigl(\id + X_{b^\star}\bigr)
=
\frac{1}{D}\Tr(\id)
=
1.
\]
Thus $\rho_{b^\star}$ is positive semidefinite with unit trace, and hence is a valid density matrix.
\end{proof}

\paragraph{State-space geometry within the family.}
The physicality condition above ensures that the construction indeed defines a valid family of density matrices. It is also useful to understand how two members of the family separate in trace distance when their hidden strings differ. If two strings agree on their first several entries, then the corresponding states contain exactly the same prefix-Pauli terms up to that depth. Only at the first location where the strings differ do the defining terms of the two states begin to diverge. The next lemma makes this observation precise.

\begin{lemma}[Trace-distance geometry of the prefix/tree family]
\label{lem:prefix_family_trace_geometry}
Let $b,\widetilde b \in \{X,Y,Z\}^n$ be two distinct hidden strings, and let
\[
r := \min \{\, i \in \{1,\dots,n\} : b_i \neq \widetilde b_i \,\}
\]
be the first location at which they differ. Then
\begin{equation}
\rho_b - \rho_{\widetilde b}
=
\frac{1}{D}
\left(
\sum_{k=r}^{n-1} \beta_k \bigl(P_b^{(k)} - P_{\widetilde b}^{(k)}\bigr)
+
\alpha \bigl(P_b^{(n)} - P_{\widetilde b}^{(n)}\bigr)
\right).
\label{eq:family_difference_first_mismatch}
\end{equation}
In follows that
\begin{equation}
\frac{1}{2}\,\|\rho_b - \rho_{\widetilde b}\|_1
\;\le\;
\sum_{k=r}^{n-1} |\beta_k| + |\alpha|,
\label{eq:family_trace_upper}
\end{equation}
and also
\begin{equation}
\frac{1}{2}\,\|\rho_b - \rho_{\widetilde b}\|_1
\;\ge\;
\frac{|c_r|}{2},
\quad
c_r :=
\begin{cases}
\beta_r, & 1 \le r \le n-1,\\
\alpha, & r=n.
\end{cases}
\label{eq:family_trace_lower}
\end{equation}
\end{lemma}

\begin{proof}
By definition,
\[
\begin{aligned}
\rho_b
&=
\frac{1}{D}
\left(
\id
+
\sum_{k=1}^{n-1}\beta_k P_b^{(k)}
+
\alpha P_b^{(n)}
\right),
\\[4pt]
\rho_{\widetilde b}
&=
\frac{1}{D}
\left(
\id
+
\sum_{k=1}^{n-1}\beta_k P_{\widetilde b}^{(k)}
+
\alpha P_{\widetilde b}^{(n)}
\right).
\end{aligned}
\]
Since $b$ and $\widetilde b$ agree in their first $r-1$ entries, we have $P_b^{(k)} = P_{\widetilde b}^{(k)}$ for $1 \le k \le r-1$. Subtracting the two state formulas therefore gives~\eqref{eq:family_difference_first_mismatch}.

For the upper bound, apply the triangle inequality:
\[
\begin{aligned}
\|\rho_b - \rho_{\widetilde b}\|_1
&\le
\frac{1}{D}
\Bigg(
\sum_{k=r}^{n-1} |\beta_k|\,\|P_b^{(k)} - P_{\widetilde b}^{(k)}\|_1
\\
&\qquad\qquad
+
|\alpha|\,\|P_b^{(n)} - P_{\widetilde b}^{(n)}\|_1
\Bigg).
\end{aligned}
\]
Each prefix Pauli operator is Hermitian unitary, so \(\|P_b^{(k)}\|_1 = \|P_{\widetilde b}^{(k)}\|_1 = D,\)
and hence
\[
\|P_b^{(k)} - P_{\widetilde b}^{(k)}\|_1
\le
\|P_b^{(k)}\|_1 + \|P_{\widetilde b}^{(k)}\|_1
=
2D.
\]
The same bound holds for the full-depth term. Therefore
\[
\|\rho_b - \rho_{\widetilde b}\|_1
\le
2\left(\sum_{k=r}^{n-1} |\beta_k| + |\alpha|\right),
\]
which yields \eqref{eq:family_trace_upper}.

For the lower bound, let
\[
Q_r :=
\begin{cases}
P_b^{(r)}, & 1 \le r \le n-1,\\
P_b^{(n)}, & r=n.
\end{cases}
\]
Since $Q_r$ is Hermitian unitary, $\|Q_r\|_\infty=1$. Using the standard inequality
\(|\Tr(AX)| \le \|A\|_\infty \|X\|_1\), we obtain \(\bigl|\Tr\bigl(Q_r(\rho_b-\rho_{\widetilde b})\bigr)\bigr|
\le
\|\rho_b-\rho_{\widetilde b}\|_1\). Hence
\[
\frac12 \|\rho_b-\rho_{\widetilde b}\|_1
\ge
\frac12 \bigl|\Tr\bigl(Q_r(\rho_b-\rho_{\widetilde b})\bigr)\bigr|.
\]
Now substitute \eqref{eq:family_difference_first_mismatch} into the trace expression.
For the depth-$r$ term, $Q_r$ coincides with the Pauli operator coming from $\rho_b$, so this contribution is nonzero. The corresponding depth-$r$ term from $\rho_{\widetilde b}$ is a different Pauli string, because $b_r \neq \widetilde b_r$, and is therefore orthogonal to $Q_r$. For every deeper term $k>r$, both $P_b^{(k)}$ and $P_{\widetilde b}^{(k)}$ are also distinct from $Q_r$, so their Hilbert--Schmidt inner products with $Q_r$ vanish as well. Thus only the first differing term contributes, which gives
\[
\Tr\bigl(Q_r(\rho_b-\rho_{\widetilde b})\bigr)=c_r.
\]
where
\[
c_r :=
\begin{cases}
\beta_r, & 1 \le r \le n-1,\\
\alpha, & r=n.
\end{cases}
\]
Therefore
\[
\frac12 \|\rho_b - \rho_{\widetilde b}\|_1
\ge
\frac{|c_r|}{2},
\]
which proves \eqref{eq:family_trace_lower}.
\end{proof}

An immediate consequence of the upper bound is that if an estimate $\widehat b$ agrees with the true hidden string $b^\star$ on its first $s$ coordinates, then 
\[
\frac{1}{2}\,\|\rho_{\widehat b}-\rho_{b^\star}\|_1
\le
\sum_{k=s+1}^{n-1} |\beta_k| + |\alpha|.
\]
Thus partial prefix recovery already implies a quantitative trace-distance guarantee.

\subsection{Why the family is hierarchical}
\label{subsec:prefix_family_hierarchy}

The defining feature of equation~\eqref{eq:prefix_family_main} is that information is distributed across all prefix depths rather than being concentrated only at the full length-$n$ string. This creates a coarse-to-fine structure over the tree of Pauli basis strings, where each node has three children corresponding to appending $X$, $Y$, or $Z$.

To make this precise, let \(b = (b_1,\dots,b_n)\in\cB\) be any allowed measurement setting, and for each depth \(k\in\{1,\dots,n\}\) define the corresponding prefix observable
\[
P_b^{(k)}
=
\left(\bigotimes_{i=1}^{k} \sigma_{b_i}\right)\otimes \id^{\otimes(n-k)}.
\]
One shot in the product-Pauli basis \(b\) produces an outcome \(o=(o_1,\dots,o_n)\in\{0,1\}^n.\) From such shots, however, one can actually compute sample values for all prefix observables \(P_b^{(k)}\).

To show this, for each site \(i\), let \(\lambda_i(o_i)\in\{\pm1\}\) denote the eigenvalue of \(\sigma_{b_i}\) associated with the local outcome \(o_i\). For each \(k\), define the scalar statistic
\begin{equation}
S_{k,b}(o):=\prod_{i=1}^k \lambda_i(o_i).
\label{eq:Skb}
\end{equation}
Then \(S_{k,b}(o)\) is exactly the eigenvalue of the prefix observable \(P_b^{(k)}\) on the product-basis outcome \(o\). Equivalently,
\[
P_b^{(k)}\Pi_o^{(b)} = S_{k,b}(o)\,\Pi_o^{(b)}.
\]
Using this identity together with \(\sum_o \Pi_o^{(b)} = \id,\)
we obtain
\begin{align*}
\mathbb{E}_\rho\!\big[S_{k,b}(O)\big]
&=
\sum_o S_{k,b}(o)\,p(o\mid b,\rho)
\\
&=
\sum_o S_{k,b}(o)\,\Tr\!\big(\rho\,\Pi_o^{(b)}\big)
\\
&=
\sum_o \Tr\!\big(\rho\,P_b^{(k)}\Pi_o^{(b)}\big)
\\
&=
\Tr\!\big(\rho\,P_b^{(k)}\sum_o \Pi_o^{(b)}\big)
\\
&=
\Tr\!\big(P_b^{(k)}\,\rho\big).
\end{align*}
Thus, for each \(k\), the expectation of the prefix observable \(P_b^{(k)}\) is estimated from outcomes obtained by measuring in the allowed basis \(b\). A single measurement setting \(b\) therefore yields sample information for all prefix observables \(P_b^{(k)}\), \(k=1,\dots,n\). The following lemma confirms that estimating these prefix observables isolates the exact structural information needed to recover the hidden string $b^\star$.

\begin{lemma}[Prefix-selective expectation pattern]
\label{lem:prefix_selective_expectation}
Fix a hidden string $b^\star \in \cB$ and define $\rho_{b^\star}$ as in~\eqref{eq:prefix_family_main}. Then for any measurement setting $b \in \cB$ and any depth $k \in \{1,\dots,n-1\}$,
\[
\Tr\!\bigl(P_b^{(k)} \rho_{b^\star}\bigr)
=
\begin{cases}
\beta_k, & \text{if } (b_1,\dots,b_k) = (b^\star_1,\dots,b^\star_k),\\
0, & \text{otherwise.}
\end{cases}
\]
At the final depth,
\[
\Tr\!\bigl(P_b^{(n)} \rho_{b^\star}\bigr)
=
\begin{cases}
\alpha, & \text{if } b = b^\star,\\
0, & \text{otherwise.}
\end{cases}
\]
\end{lemma}

\begin{proof}
We prove the claim by expanding the trace and using Pauli orthogonality. Fix $b \in \cB$ and $k \in \{1,\dots,n-1\}$. Write $P := P_b^{(k)}$. Then
\[
\begin{aligned}
\Tr(P\rho_{b^\star})
&=
\frac{1}{D}
\Bigl(
\Tr(P)
+
\sum_{j=1}^{n-1}\beta_j \Tr\!\bigl(P P_{b^\star}^{(j)}\bigr)
\Bigr)
\\
&\qquad
+
\frac{\alpha}{D}\Tr\!\bigl(P P_{b^\star}^{(n)}\bigr).
\end{aligned}
\]
First, \(\Tr(P)=0,\) because $P_b^{(k)}$ is a non-identity Pauli string. Next consider the terms $\Tr\!\bigl(P P_{b^\star}^{(j)}\bigr)$ for
$j=1,\dots,n-1$. If $j \neq k$, then the two strings have different supports or different lengths, so their product is again a non-identity Pauli string. Hence $\Tr\!\bigl(P P_{b^\star}^{(j)}\bigr)=0$. Thus the only potentially nonzero term in the sum is the one with $j=k$. For that term,
\[
\Tr\!\bigl(P_b^{(k)} P_{b^\star}^{(k)}\bigr)
=
\begin{cases}
D, & \text{if } (b_1,\dots,b_k)=(b^\star_1,\dots,b^\star_k),\\
0, & \text{otherwise,}
\end{cases}
\]
because two Pauli strings have Hilbert--Schmidt inner product equal to $D$ when they are identical and $0$ when they are distinct. Since $k<n$, the product \(P_b^{(k)} P_{b^\star}^{(n)}\) is also a non-identity Pauli string, so \(\Tr\!\bigl(P_b^{(k)} P_{b^\star}^{(n)}\bigr)=0.\) Combining these facts gives
\[
\begin{aligned}
\Tr\!\bigl(P_b^{(k)} \rho_{b^\star}\bigr)
&=
\frac{1}{D}
\left(
\beta_k \Tr\!\bigl(P_b^{(k)} P_{b^\star}^{(k)}\bigr)
\right)
\\
&=
\begin{cases}
\beta_k, & \text{if } (b_1,\dots,b_k)=(b^\star_1,\dots,b^\star_k),\\
0, & \text{otherwise.}
\end{cases}
\end{aligned}
\]

Now consider the final depth $k=n$. Then \(P_b^{(n)}\) is a full Pauli string, and similarly
\[
\begin{aligned}
\Tr\!\bigl(P_b^{(n)} \rho_{b^\star}\bigr)
&=
\frac{1}{D}
\Bigl(
\Tr(P_b^{(n)})
+
\sum_{j=1}^{n-1}\beta_j \Tr\!\bigl(P_b^{(n)} P_{b^\star}^{(j)}\bigr)
\Bigr)
\\
&\qquad
+
\frac{\alpha}{D}\Tr\!\bigl(P_b^{(n)} P_{b^\star}^{(n)}\bigr).
\end{aligned}
\]
Again, \(\Tr(P_b^{(n)})=0,\) and for every $j<n$, \(\Tr\!\bigl(P_b^{(n)} P_{b^\star}^{(j)}\bigr)=0\) because the product is a non-identity Pauli string. Therefore only the final term can contribute:
\[
\Tr\!\bigl(P_b^{(n)} P_{b^\star}^{(n)}\bigr)
=
\begin{cases}
D, & \text{if } b=b^\star,\\
0, & \text{otherwise.}
\end{cases}
\]
Hence
\[
\Tr\!\bigl(P_b^{(n)} \rho_{b^\star}\bigr)
=
\begin{cases}
\alpha, & \text{if } b=b^\star,\\
0, & \text{otherwise.}
\end{cases}
\]
This proves the lemma.
\end{proof}

Lemma~\ref{lem:prefix_selective_expectation} is the mathematical form of the breadcrumb intuition. A measurement basis does not need to match the entire hidden string in order to become informative. It is already useful as soon as it matches the correct prefix at some depth. In that case, the expectation $\Tr(P_b^{(k)}\rho_{b^\star})$ is nonzero and equals $\beta_k$. (As we will see in Section~\ref{subsec:prefix_family_spike_limit}, this is the crucial difference between the present family and an all-or-nothing spike construction.)

A useful way to visualize the family is through the tree of basis strings in $\{X,Y,Z\}^n$, where each node has three children corresponding to appending $X$, $Y$, or $Z$. The unknown string $b^\star$ determines a single root-to-leaf path in this tree. The coefficient $\beta_1$ is associated with the first step along that path, $\beta_2$ with the second, and more generally $\beta_k$ with depth $k$, while $\alpha$ is associated with the final depth-$n$ observable. A measurement basis that reaches the correct node at depth $k$ already reveals the corresponding depth-$k$ signal, even if the remaining symbols do not match. In this way, the hidden string can be identified progressively, one level at a time.

\subsection{Why this helps adaptivity}
\label{subsec:prefix_family_adaptivity_intuition}

The adaptive mechanism suggested by the family is now straightforward. Suppose an adaptive protocol has already identified the first $k-1$ symbols of the hidden string. To determine the next symbol, it only needs to compare the three candidate extensions \(a \in \{X,Y,Z\}\).

For each candidate $a$, one measures in a basis whose first $k$ symbols agree with the known prefix followed by $a$. By Lemma~\ref{lem:prefix_selective_expectation}, the corresponding depth-$k$ expectation is equal to $\beta_k$ for the correct candidate and $0$ for the two incorrect candidates. Thus each stage reduces to a constant-size hypothesis test with a nonzero signal gap. Repeating this argument from $k=1$ through $k=n-1$ reveals the hidden prefix sequentially, and, once the first $n-1$ symbols have been identified, the final stage uses the full-length coefficient $\alpha$ to determine the last symbol.

This stagewise viewpoint shows how the global search over $3^n$ possible basis strings can be replaced by a sequence of $n$ local three-way decisions, one at each depth of the tree. This mechanism powers the adaptive upper bound proved in the next section, where we also provide the formal concentration analysis bounding the sample complexity and error probability of each stage.

\subsection{Relation to the spike-like limit}
\label{subsec:prefix_family_spike_limit}

It is also helpful to note what happens when the intermediate coefficients are removed. If $\beta_k = 0$ for all $k = 1,\dots,n-1$, then equation~\eqref{eq:prefix_family_main} reduces to
\[
\rho_{b^\star}
=
\frac{1}{D}
\left(
\id + \alpha P_{b^\star}^{(n)}
\right).
\]
In this limit, only the full basis string is informative. The hierarchical breadcrumb structure disappears, and the family reduces to an all-or-nothing Pauli-tilt construction. This makes clear that the intermediate prefix terms are what enable sequential localization.

We now formalize this mechanism: Section~\ref{sec:adaptive_upper} develops the constructive adaptive recovery procedure, while Section~\ref{sec:nonadaptive_lower} shows how this same structure yields a worst-case obstruction for any non-adaptive allocation.

\section{Adaptive Upper Bound}
\label{sec:adaptive_upper}

We now turn the breadcrumb mechanism from Section~\ref{sec:prefix_family} into a quantitative adaptive recovery guarantee. The adaptive measurement protocol proceeds sequentially. For $k=1,\dots,n$, once the first $k-1$ symbols have been estimated as $\widehat{b}_1, \dots, \widehat{b}_{k-1}$, the $k$th symbol can be recovered by comparing the three possible next Pauli choices. To do this, we construct a measurement setting $b \in \cB$ whose first $k-1$ symbols match the estimated prefix: $(b_1,\dots,b_{k-1}) = (\widehat{b}_1,\dots,\widehat{b}_{k-1})$. For the $k$th position $b_k$, we test each candidate $a\in\{X,Y,Z\}$. When $k < n$, the choice of the remaining symbols $b_{k+1},\dots,b_n$ does not affect the depth-$k$ prefix observable; for concreteness, we may simply pad with $X$, setting $b_{k+1} = \dots = b_n = X$.

For each candidate $a \in \{X,Y,Z\}$, we collect $m_k$ i.i.d.\ samples in the corresponding basis $b$. Let $O^{(t,a)} \in \{0,1\}^n$ denote the outcome of the $t$-th measurement shot in that basis. Using~\eqref{eq:Skb}, we compute the empirical mean of the prefix statistic,
\[
\widehat s_{k,a} = \frac{1}{m_k}\sum_{t=1}^{m_k} S_{k,b}\big(O^{(t,a)}\big),
\]
and select the next symbol by maximizing the absolute signal:
\[
\widehat b_k \in \arg\max_{a\in\{X,Y,Z\}} |\widehat s_{k,a}|.
\]
We now proceed to prove Theorem~\ref{thm:adaptive_prefix}.

\begin{proof}[Proof of Theorem~\ref{thm:adaptive_prefix}]
Fix $k\in\{1,\dots,n\}$, and condition on the event that the prefix \((b^\star_1,\dots,b^\star_{k-1})\) has already been identified correctly. By Lemma~\ref{lem:prefix_selective_expectation},
\[
\begin{aligned}
\bbE[\widehat s_{k,a}]
&=
\begin{cases}
\mu_k, & a=b^\star_k,\\
0, & a\neq b^\star_k,
\end{cases}
\\[6pt]
\mu_k
&:=
\begin{cases}
\beta_k, & k\in\{1,\dots,n-1\},\\
\alpha, & k=n.
\end{cases}
\end{aligned}
\]
Define
\[
\begin{aligned}
\mathcal E_k
:={}&
\Big\{\,|\widehat s_{k,b^\star_k}-\mu_k|
\le \tfrac12|\mu_k|\,\Big\}
\\
&{}\cap
\bigcap_{a\neq b^\star_k}
\Big\{\,|\widehat s_{k,a}|
\le \tfrac12|\mu_k|\,\Big\}.
\end{aligned}
\]
On $\mathcal E_k$, $|\widehat s_{k,b^\star_k}| \ge \tfrac12|\mu_k|$ and $|\widehat s_{k,a}| \le \tfrac12|\mu_k|$ for all $a\neq b^\star_k$, so the argmax rule selects the correct symbol: $\widehat b_k=b^\star_k$.

Since $S_{k,a}(O^{(t,a)})\in\{\pm1\}\subset[-1,1]$, Hoeffding's inequality \cite{hoeffding1963probability} gives
\[
\bbP\!\left(
|\widehat s_{k,a}-\bbE[\widehat s_{k,a}]|
\ge \tfrac12|\mu_k|
\right)
\le
2\exp\!\left(-\frac{m_k\mu_k^2}{8}\right)
\]
for each candidate $a$. A union bound over the three candidates yields
\[
\bbP(\mathcal E_k^c)
\le
6\exp\!\left(-\frac{m_k\mu_k^2}{8}\right).
\]
Therefore, if
\[
m_k \ge \frac{8}{\mu_k^2}\log\!\Big(\frac{6n}{\eta}\Big),
\]
then $\bbP(\mathcal E_k^c)\le \frac{\eta}{n}$.

A union bound over $k=1,\dots,n$ gives $\bbP\big((\widehat b_1,\dots,\widehat b_n)=(b^\star_1,\dots,b^\star_n)\big)\ge 1-\eta$.

Finally, since each stage tests three candidates, $M_{\mathrm{ad}} = 3\sum_{k=1}^{n} m_k$. Using the above bound on $m_k$,
\[
M_{\mathrm{ad}}
\le
24\log\!\Big(\frac{6n}{\eta}\Big)
\left(
\sum_{k=1}^{n-1}\beta_k^{-2}+\alpha^{-2}
\right).
\]
\end{proof}

As an immediate consequence, the same adaptive procedure also yields a trace-distance tomography guarantee over the family.

\begin{corollary}[Adaptive tomography guarantee over the prefix/tree family]
\label{cor:adaptive_prefix_tomography}
Assume the setting of Theorem~\ref{thm:adaptive_prefix}, and let $\widehat b$ denote the hidden-string estimate returned by the adaptive prefix-recovery procedure. Define the state estimator by
\[
\widehat\rho := \rho_{\widehat b}.
\]
Then
\begin{equation}
\sup_{\rho_{b^\star}\in \cF_n(\alpha,\beta)}
\bbP_{\rho_{b^\star}}\!\left[
\frac12 \|\widehat\rho-\rho_{b^\star}\|_1 > 0
\right]
\le \eta.
\label{eq:adaptive_exact_tomo_cor}
\end{equation}
Consequently, for every $\varepsilon>0$, the same adaptive procedure is $(\varepsilon,\eta)$-accurate over $\cF_n(\alpha,\beta)$, i.e.,
\begin{equation}
\sup_{\rho_{b^\star}\in \cF_n(\alpha,\beta)}
\bbP_{\rho_{b^\star}}\!\left[
\frac12 \|\widehat\rho-\rho_{b^\star}\|_1 > \varepsilon
\right]
\le \eta.
\label{eq:adaptive_epsilon_tomo_cor}
\end{equation}
\end{corollary}

\begin{proof}
By Theorem~\ref{thm:adaptive_prefix}, with probability at least $1-\eta$, the adaptive procedure returns the correct hidden string, that is, $\widehat b = b^\star$. In this event, the estimator satisfies $\widehat\rho = \rho_{\widehat b} = \rho_{b^\star}$, and therefore $\frac12 \|\widehat\rho-\rho_{b^\star}\|_1 = 0$. Hence the event $\left\{
\frac12 \|\widehat\rho-\rho_{b^\star}\|_1 > 0
\right\}$ is contained in the event $\{\widehat b \neq b^\star\}$. The bound
\eqref{eq:adaptive_exact_tomo_cor} therefore follows directly from the recovery guarantee in Theorem~\ref{thm:adaptive_prefix}. Since, for every $\varepsilon>0$,
\[
\left\{
\frac12 \|\widehat\rho-\rho_{b^\star}\|_1 > \varepsilon
\right\}
\subseteq
\left\{
\frac12 \|\widehat\rho-\rho_{b^\star}\|_1 > 0
\right\},
\]
the bound \eqref{eq:adaptive_epsilon_tomo_cor} follows immediately as well.
\end{proof}

We now return to the copy-complexity bound and instantiate it with the concrete coefficient choice from Corollary~\ref{cor:adaptive_nonadaptive_separation}.

\begin{proof}[Proof of Corollary~\ref{cor:adaptive_nonadaptive_separation} statement $(i)$]
Under the choice
\[
\alpha=\frac{\varepsilon}{4},
\qquad
\beta_k=\frac{\varepsilon}{4(n-1)},
\qquad k=1,\dots,n-1,
\]
the theorem gives
\[
\begin{aligned}
m_k
&\ge
\frac{8}{\beta_k^2}\log\!\Big(\frac{6n}{\eta}\Big)
\\
&=
\frac{128(n-1)^2}{\varepsilon^2}\log\!\Big(\frac{6n}{\eta}\Big),
\qquad
k=1,\dots,n-1.
\end{aligned}
\]
and
\[
m_n
\ge
\frac{8}{\alpha^2}\log\!\Big(\frac{6n}{\eta}\Big)
=
\frac{128}{\varepsilon^2}\log\!\Big(\frac{6n}{\eta}\Big).
\]
Hence
\[
\begin{aligned}
M_{\mathrm{ad}}
&=
3\sum_{k=1}^{n}m_k
\\
&\le
\frac{384}{\varepsilon^2}
\log\!\Big(\frac{6n}{\eta}\Big)\big((n-1)^3+1\big)
\\
&=
\tilde O\!\big(n^3\varepsilon^{-2}\big).
\end{aligned}
\]
\end{proof}

Thus, for this family, adaptivity reduces recovery to a sequence of local three-way tests and achieves the stated copy-complexity guarantee. We next contrast this with the non-adaptive setting.

\section{Non-Adaptive Lower Bound}
\label{sec:nonadaptive_lower}

We now turn to the proof of Theorem~\ref{thm:nonad_lb_full}. Fix an arbitrary non-adaptive design over $\cB$. For each prefix $u\in\{X,Y,Z\}^{n-1}$, define
\[
\cB_u:=\{(u,X),(u,Y),(u,Z)\},
\]
the three measurement settings that share the first $n-1$ coordinates specified by $u$. We refer to $\cB_u$ as the prefix cylinder associated with $u$. The proof proceeds in three steps. We first identify a prefix cylinder that receives only a small fraction of the measurement budget. We then construct two states that agree on the common prefix and differ only in the final coordinate. This choice implies that measurement settings outside the selected cylinder induce the same one-shot distribution under both states, so any distinguishing information is concentrated on the small amount of budget assigned to that cylinder.

\begin{lemma}[A rarely sampled prefix cylinder]
\label{lem:rare_cylinder}
Let
\[
\mathcal{U}_{n-1}:=\{X,Y,Z\}^{n-1},
\]
\[
\cB_u:=\{(u,X),(u,Y),(u,Z)\}\subset \cB,
\quad \text{for } u\in\mathcal{U}_{n-1}.
\]
Then $\{\cB_u : u\in\mathcal{U}_{n-1}\}$ is a partition of $\cB$. Moreover, for every non-adaptive design there exists a prefix \(u^\star\in\mathcal{U}_{n-1}\) such that
\begin{equation}
\mu(\cB_{u^\star})\le 3^{-(n-1)}.
\label{eq:rare_cylinder}
\end{equation}
\end{lemma}

\begin{proof}
The sets $\cB_u$ are disjoint and their union is all of $\cB=\{X,Y,Z\}^n$, so they form a partition. Let $\mu$ denote the allocation induced by the non-adaptive design on $\cB$, with $\mu(b)\ge 0$ and $\sum_{b\in\cB}\mu(b)=1$. For a deterministic $M$-shot design, $\mu(b)$ is the empirical fraction of times that setting $b$ is used; for a randomized non-adaptive design, $\mu(b)$ is the corresponding expected fraction. Therefore
\[
\sum_{u\in\mathcal{U}_{n-1}} \mu(\cB_u)=1.
\]
Since $|\mathcal{U}_{n-1}|=3^{n-1}$, the pigeonhole principle yields a prefix $u^\star$ such that $\mu(\cB_{u^\star})\le 3^{-(n-1)}$.
\end{proof}

We now fix such a rare prefix $u^\star$ for the given non-adaptive design.

\begin{lemma}[Hard pair with prefix cancellation]
\label{lem:hard_pair_prefix}
Define the hard basis string pair
\[
b^{(0)}:=(u^\star,X),
\qquad
b^{(1)}:=(u^\star,Y)
\]
and corresponding density matrices $\rho^{(0)}:=\rho_{b^{(0)}}$ and $\rho^{(1)}:=\rho_{b^{(1)}}$.
Then
\begin{equation}
\rho^{(0)}-\rho^{(1)}
=
\frac{\alpha}{D}\Big(P^{(n)}_{b^{(0)}}-P^{(n)}_{b^{(1)}}\Big).
\label{eq:rho_diff_two_point}
\end{equation}
Moreover,
\begin{equation}
\frac12\|\rho^{(0)}-\rho^{(1)}\|_1
=
\frac{|\alpha|}{\sqrt{2}}.
\label{eq:two_point_trdist}
\end{equation}
\end{lemma}

\begin{proof}
Because $b^{(0)}$ and $b^{(1)}$ share the same prefix of length $n-1$, we have $P^{(k)}_{b^{(0)}} = P^{(k)}_{b^{(1)}} = P^{(k)}_{u^\star}$ for $k=1,\dots,n-1$. Thus all prefix contributions cancel, which gives \eqref{eq:rho_diff_two_point}.

Since the two full strings differ only in the final coordinate, there exists a common $(n-1)$-qubit Pauli string $Q$ such that $P^{(n)}_{b^{(0)}} = Q \otimes \sigma_X$ and
$P^{(n)}_{b^{(1)}} = Q \otimes \sigma_Y$. Hence $P^{(n)}_{b^{(0)}} - P^{(n)}_{b^{(1)}}
=
Q \otimes (\sigma_X-\sigma_Y)$. Using \eqref{eq:rho_diff_two_point}, we obtain
\[
\frac12\|\rho^{(0)}-\rho^{(1)}\|_1
=
\frac{|\alpha|}{2D}\,
\bigl\|Q\otimes(\sigma_X-\sigma_Y)\bigr\|_1.
\]

We now evaluate the two trace norms. Since $Q$ is an $(n-1)$-qubit Pauli string, it is a Hermitian unitary with eigenvalues $\pm 1$. Therefore all of its singular values are equal to $1$, and $\|Q\|_1 = 2^{n-1}$. Also, $(\sigma_X-\sigma_Y)^2 = 2\id$, so the eigenvalues of $\sigma_X-\sigma_Y$ are $\pm \sqrt{2}$. Hence $\|\sigma_X-\sigma_Y\|_1 = 2\sqrt{2}$. Using the multiplicativity of the trace norm over tensor products, $\|A\otimes B\|_1=\|A\|_1\|B\|_1$, we conclude that
\[
\bigl\|Q\otimes(\sigma_X-\sigma_Y)\bigr\|_1
=
\|Q\|_1\,\|\sigma_X-\sigma_Y\|_1
=
2^{n-1}\cdot 2\sqrt{2}.
\]
Since $D=2^n$, it follows that
\[
\frac12\|\rho^{(0)}-\rho^{(1)}\|_1
=
\frac{|\alpha|}{2\cdot 2^n}\,(2^{n-1})(2\sqrt{2})
=
\frac{|\alpha|}{\sqrt{2}},
\]
which proves \eqref{eq:two_point_trdist}.
\end{proof}

We thus choose the two hard instances to agree on the first $n-1$ coordinates and differ only in the last. As a result, all shorter-prefix terms are identical under the two states, and only the full-depth term contributes to their difference. A key fact in our analysis is that any measurement basis $b$ whose prefix is not $u^\star$ will not aid in distinguishing $\rho^{(0)}$ from $\rho^{(1)}$.

To show this, we write the common prefix-only part as
\begin{equation}
\rho_{\mathrm{pref}}
:=
\frac{1}{D}
\left(
\id+\sum_{k=1}^{n-1}\beta_k P^{(k)}_{u^\star}
\right).
\label{eq:rho_pref_def}
\end{equation}
Then
\begin{equation}
\rho^{(j)}
=
\rho_{\mathrm{pref}}+\frac{\alpha}{D}P^{(n)}_{b^{(j)}},
\qquad
j\in\{0,1\}.
\label{eq:rho_pref_plus_full}
\end{equation}
We are now ready to prove the following key result.

\begin{lemma}[Outside the selected prefix cylinder, the one-shot laws coincide]
\label{lem:outside_cylinder_coincide}
If $b\notin \cB_{u^\star}$, then
\[
p(\cdot\mid b,\rho^{(0)})=p(\cdot\mid b,\rho^{(1)}).
\]
Consequently,
\begin{equation}
D_{\KL}\!\big(p(\cdot\mid b,\rho^{(0)})\,\|\,p(\cdot\mid b,\rho^{(1)})\big)=0.
\label{eq:outside_cylinder_zero_KL}
\end{equation}
Moreover, the same conclusion holds for the third basis inside the cylinder, \(b=(u^\star,Z).\)
\end{lemma}

\begin{proof}
Fix a basis $b=(b_1,\dots,b_n)\in\cB$ and an outcome $o=(o_1,\dots,o_n)\in\{0,1\}^n$. Since the measurement is a product-Pauli basis measurement, the corresponding product measurement projector factors as
\[
\Pi_o^{(b)}=\bigotimes_{i=1}^n \pi_{o_i}^{(b_i)},
\]
where each local projector is $\pi_{o_i}^{(b_i)}:=\frac12\bigl(\id+(-1)^{o_i}\sigma_{b_i}\bigr)$ for $o_i\in\{0,1\}$. Thus $\pi_{0}^{(b_i)}$ and $\pi_{1}^{(b_i)}$ are the two rank-one eigenprojectors of $\sigma_{b_i}$.

Now let
\[
P_{b'}^{(n)}=\bigotimes_{i=1}^n \sigma_{b'_i}
\]
be any full Pauli string. By multiplicativity of the trace over tensor products,
\[
\Tr\!\big(\Pi_o^{(b)}P_{b'}^{(n)}\big)
=
\prod_{i=1}^n \Tr\!\big(\pi_{o_i}^{(b_i)}\sigma_{b'_i}\big).
\]
If $b\neq b'$, then there exists a coordinate $i_0$ such that $b_{i_0}\neq b'_{i_0}$. Because $\pi_{o_{i_0}}^{(b_{i_0})}$ is an eigenprojector of $\sigma_{b_{i_0}}$, its expectation against the different Pauli $\sigma_{b'_{i_0}}$ is zero. Hence $\Tr\!\big(\Pi_o^{(b)} P_{b'}^{(n)}\big)=0$
 for all $o$ whenever $b\neq b'$.

Now suppose $b\notin \cB_{u^\star}$. Then $b$ disagrees with $u^\star$ in at least one of the first $n-1$ coordinates, so it cannot equal either $b^{(0)}=(u^\star,X)$ or $b^{(1)}=(u^\star,Y)$. Therefore
\[
\Tr\!\big(\Pi_o^{(b)} P^{(n)}_{b^{(0)}}\big)
=
\Tr\!\big(\Pi_o^{(b)} P^{(n)}_{b^{(1)}}\big)
=
0
\qquad\text{for all }o.
\]
Using \eqref{eq:rho_pref_plus_full},
\[
\begin{aligned}
p(o\mid b,\rho^{(j)})
&=
\Tr\!\big(\Pi_o^{(b)}\rho_{\mathrm{pref}}\big)
+
\frac{\alpha}{D}\Tr\!\big(\Pi_o^{(b)} P^{(n)}_{b^{(j)}}\big)
\\
&=
\Tr\!\big(\Pi_o^{(b)}\rho_{\mathrm{pref}}\big).
\end{aligned}
\]
for both $j=0,1$. Hence the one-shot laws are identical, proving \eqref{eq:outside_cylinder_zero_KL}.

Finally, if $b=(u^\star,Z)$, then again $b\neq b^{(0)}$ and $b\neq b^{(1)}$, so the same argument applies.
\end{proof}

\begin{lemma}[Transcript KL is bounded by the allocation on the selected prefix cylinder]
\label{lem:kl_contract_nonad}
Let
\[
T=(B_1,O_1,\dots,B_M,O_M)
\]
be the full transcript, and let $P^{(j)}_{\mu,M}$ denote the distribution of $T$ when the true state is $\rho^{(j)}$ and the non-adaptive design $\mu$ is used. Suppose there exists a constant $\delta\in(0,1)$ such that for each of the two settings \(b\in\{b^{(0)},b^{(1)}\}\) and every outcome $o\in\{0,1\}^n$,
\begin{equation}
p(o\mid b,\rho_{\mathrm{pref}})\ge \delta\,2^{-n}.
\label{eq:baseline_lower_bound}
\end{equation}
Then there exists an absolute constant $C<\infty$ such that
\begin{equation}
\begin{aligned}
D_{\KL}\!\big(P^{(0)}_{\mu,M}\,\|\,P^{(1)}_{\mu,M}\big)
&\le
C\,M\,\mu(\cB_{u^\star})\,\mathrm{kl}(\alpha)
\\
&\le
\frac{C\,M}{3^{n-1}}\,\mathrm{kl}(\alpha).
\end{aligned}
\label{eq:transcript_KL_contraction}
\end{equation}
where $\mathrm{kl}(\alpha)$ is as defined in Theorem~\ref{thm:nonad_lb_full}.
\end{lemma}

\begin{proof}
Under a non-adaptive design, the measurement settings are selected in advance and do not depend on the observed outcomes. Therefore, the KL chain rule gives
\begin{multline}
D_{\KL}\!\big(P^{(0)}_{\mu,M}\,\|\,P^{(1)}_{\mu,M}\big)
=
\sum_{t=1}^M
\mathbb{E}\!\Bigl[
D_{\KL}\!\Bigl(
p(\cdot\mid B_t,\rho^{(0)})
\\
\|\, p(\cdot\mid B_t,\rho^{(1)})
\Bigr)
\Bigr].
\label{eq:chain_rule_nonad}
\end{multline}
By Lemma~\ref{lem:outside_cylinder_coincide}, the per-shot KL term is exactly zero whenever \(B_t\notin \cB_{u^\star},\) and it is also zero when \(B_t=(u^\star,Z).\) Thus only the two settings $b^{(0)}=(u^\star,X)$ and
$b^{(1)}=(u^\star,Y)$ can contribute.

It therefore remains to bound the one-shot KL divergence for these two settings. We present the calculation for $B_t=b^{(0)}$, since the same argument applies to $B_t=b^{(1)}$. For each outcome $o\in\{0,1\}^n$, define
\[
q(o):=p(o\mid b^{(0)},\rho_{\mathrm{pref}})
=\Tr\!\big(\Pi_o^{(b^{(0)})}\rho_{\mathrm{pref}}\big).
\]
Since \eqref{eq:baseline_lower_bound} holds for each $b\in\{b^{(0)},b^{(1)}\}$, in particular for $b=b^{(0)}$ we have $q(o)\ge \delta\,2^{-n}$. Using \eqref{eq:rho_pref_plus_full}, we can write
\[
\begin{aligned}
p(o\mid b^{(0)},\rho^{(j)})
&=
q(o)
+
\frac{\alpha}{D}
\Tr\!\big(\Pi_o^{(b^{(0)})}P_{b^{(j)}}^{(n)}\big),
\\
&\qquad j\in\{0,1\}.
\end{aligned}
\]
Now set
\[
s(o):=\Tr\!\big(\Pi_o^{(b^{(0)})}P_{b^{(0)}}^{(n)}\big)\in\{+1,-1\},
\]
which is the eigenvalue of the full Pauli string $P_{b^{(0)}}^{(n)}$ associated with outcome $o$ in the basis $b^{(0)}$.

Since $b^{(1)}\neq b^{(0)}$, the two strings differ in at least one coordinate. Hence one local factor in the product expansion of \(\Tr\!\big(\Pi_o^{(b^{(0)})}P_{b^{(1)}}^{(n)}\big)\) is zero, and therefore $\Tr\!\big(\Pi_o^{(b^{(0)})}P_{b^{(1)}}^{(n)}\big)=0$. It follows that $p(o\mid b^{(0)},\rho^{(0)}) = q(o) + \frac{\alpha}{D}s(o)$ while
$p(o\mid b^{(0)},\rho^{(1)}) = q(o)$. Define
$r(o):=\frac{\alpha s(o)}{D q(o)}$. Using $|s(o)|=1$, $D=2^n$, and the lower bound on $q(o)$ above, we obtain
\[
|r(o)|
=
\frac{|\alpha||s(o)|}{Dq(o)}
\le
\frac{|\alpha|}{2^n \, \delta\,2^{-n}}
=
\frac{|\alpha|}{\delta}.
\]
Hence, for sufficiently small $|\alpha|$, we may assume $|r(o)|\le 1/2$ for all $o$.

Now define
\[
\begin{aligned}
p_0(o)
&:= p(o\mid b^{(0)},\rho^{(0)})
= q(o)\bigl(1+r(o)\bigr),
\\
p_1(o)
&:= p(o\mid b^{(0)},\rho^{(1)})
= q(o).
\end{aligned}
\]
Since both $p_0$ and $p_1$ are probability distributions on $\{0,1\}^n$, we have $\sum_o \bigl(p_0(o)-p_1(o)\bigr)=0$. Using the definitions of $p_0$ and $p_1$, this becomes $\sum_o q(o)\,r(o)=0$.

By the definition of KL divergence,
\begin{multline*}
D_{\KL}(p_0\|p_1)
=
\sum_o p_0(o)\log\!\frac{p_0(o)}{p_1(o)}
\\
=
\sum_o q(o)\bigl(1+r(o)\bigr)\log\!\bigl(1+r(o)\bigr).
\end{multline*}
Using $\sum_o q(o)r(o)=0$, we rewrite this as
\[
D_{\KL}(p_0\|p_1)
=
\sum_o q(o)\Bigl[\bigl(1+r(o)\bigr)\log\!\bigl(1+r(o)\bigr)-r(o)\Bigr].
\]
Since $|r(o)|\le 1/2$, we have in particular $r(o)>-1$. Hence we may apply the classical inequality $\log(1+x)\le x$ when
$x>-1$, which follows from the concavity of the logarithm and implies $(1+x)\log(1+x)-x
\le
(1+x)x-x
=
x^2$. Applying this with $x=r(o)$, we obtain
\begin{align*}
D_{\KL}\!\big(p_0\|p_1\big)
&\le
\sum_o q(o)\,r(o)^2
\\
&=
\sum_o \frac{\alpha^2}{D^{2}q(o)}
\\
&\le
\sum_o \frac{\alpha^2}{2^{2n}\,\delta\,2^{-n}}
\\
&=
\frac{\alpha^2}{\delta}.
\end{align*}
Since \(\mathrm{kl}(\alpha)=\Theta(\alpha^2)\) as \(\alpha\to 0\), the bound above implies
\[
D_{\KL}\!\big(p(\cdot\mid b^{(0)},\rho^{(0)})\,\|\,p(\cdot\mid b^{(0)},\rho^{(1)})\big)
\le
C\,\mathrm{kl}(\alpha)
\]
for some constant $C<\infty$ and all sufficiently small $|\alpha|$. The same bound holds for $B_t=b^{(1)}$.

Substituting these per-shot bounds into \eqref{eq:chain_rule_nonad} yields
\[
\begin{aligned}
D_{\KL}\!\big(P^{(0)}_{\mu,M}\,\|\,P^{(1)}_{\mu,M}\big)
&\le
C\sum_{t=1}^M
\mathbb{P}\!\big(B_t\in\{b^{(0)},b^{(1)}\}\big)\,\mathrm{kl}(\alpha)
\\
&\le
C\,M\,\mu(\cB_{u^\star})\,\mathrm{kl}(\alpha).
\end{aligned}
\]
The bound \eqref{eq:rare_cylinder} then gives the second inequality in \eqref{eq:transcript_KL_contraction}.
\end{proof}

The following lemma establishes that reliably distinguishing the two hard cases requires a sufficiently large KL divergence between the corresponding measurement distributions.

\begin{lemma}[Testing lower bound from KL divergence]
\label{lem:testing_KL_lower}
Let $P$ and $Q$ be distributions on the same measurable space, and let a possibly randomized test decide between
\[
H_0:P
\qquad\text{and}\qquad
H_1:Q.
\]
Let $A$ be the acceptance region for $H_1$, and define the type-I and type-II errors by $\alpha_{\mathrm{I}}:=P(A)$ and
$\alpha_{\mathrm{II}}:=Q(A^c)$. Then
\begin{equation}
D_{\KL}(P\|Q)
\ge
D_{\KL}\!\big(\mathrm{Bern}(\alpha_{\mathrm{I}})\,\|\,\mathrm{Bern}(1-\alpha_{\mathrm{II}})\big).
\label{eq:testing_KL_dp}
\end{equation}
In particular, if $\alpha_{\mathrm{I}}\le \eta$ and $\alpha_{\mathrm{II}}\le \eta$ for $\eta\in(0,1/4]$, then
\begin{equation}
D_{\KL}(P\|Q)\gtrsim \log\!\Big(\frac{1}{\eta}\Big).
\label{eq:testing_KL_log_eta}
\end{equation}
\end{lemma}

\begin{proof}
Let $\Omega$ denote the set of all possible transcripts, and let $A\subseteq \Omega$ be the acceptance region for $H_1$. With $\alpha_{\mathrm{I}}=P(A)$ and
$\alpha_{\mathrm{II}}=Q(A^c)$, note that $P(A^c)=1-\alpha_{\mathrm{I}}$ and
$Q(A)=1-\alpha_{\mathrm{II}}$.

By definition,
\[
D_{\KL}(P\|Q)
=
\sum_{t\in\Omega} P(t)\log\!\frac{P(t)}{Q(t)}.
\]
Splitting the sum over $A$ and $A^c$, we obtain
\[
D_{\KL}(P\|Q)
=
\sum_{t\in A} P(t)\log\!\frac{P(t)}{Q(t)}
+
\sum_{t\in A^c} P(t)\log\!\frac{P(t)}{Q(t)}.
\]
Applying the log-sum inequality on $A$ gives
\[
\sum_{t\in A} P(t)\log\!\frac{P(t)}{Q(t)}
\ge
P(A)\log\!\frac{P(A)}{Q(A)},
\]
and applying the same inequality on $A^c$ gives
\[
\sum_{t\in A^c} P(t)\log\!\frac{P(t)}{Q(t)}
\ge
P(A^c)\log\!\frac{P(A^c)}{Q(A^c)}.
\]
Therefore
\[
D_{\KL}(P\|Q)
\ge
P(A)\log\!\frac{P(A)}{Q(A)}
+
P(A^c)\log\!\frac{P(A^c)}{Q(A^c)}.
\]
Setting $p:=P(A)$ and $q:=Q(A)$, the right-hand side becomes
\[
p\log\!\frac{p}{q}
+
(1-p)\log\!\frac{1-p}{1-q}.
\]
This is exactly $D_{\KL}\!\big(\mathrm{Bern}(p)\,\|\,\mathrm{Bern}(q)\big)$, the KL divergence between two Bernoulli distributions. Substituting $p=\alpha_{\mathrm{I}}$ and $q=1-\alpha_{\mathrm{II}}$,
we obtain $D_{\KL}(P\|Q)
\ge
D_{\KL}\!\big(\mathrm{Bern}(\alpha_{\mathrm{I}})\,\|\,\mathrm{Bern}(1-\alpha_{\mathrm{II}})\big)$, which proves \eqref{eq:testing_KL_dp}.

For the second claim, combine \eqref{eq:testing_KL_dp} with the formula for the KL divergence between two Bernoulli distributions:
\begin{align*}
D_{\KL}(P\|Q)
&\ge
D_{\KL}\!\big(
\mathrm{Bern}(\alpha_{\mathrm{I}})
\,\|\,
\mathrm{Bern}(1-\alpha_{\mathrm{II}})
\big)
\\
&=
\alpha_{\mathrm{I}}
\log\!\frac{\alpha_{\mathrm{I}}}{1-\alpha_{\mathrm{II}}} +
(1-\alpha_{\mathrm{I}})
\log\!\frac{1-\alpha_{\mathrm{I}}}{\alpha_{\mathrm{II}}}.
\end{align*}
Since $\alpha_{\mathrm{I}},\alpha_{\mathrm{II}}\le \eta < \frac{1}{2}$, we have
\[
\alpha_{\mathrm{I}}
\log\!\frac{\alpha_{\mathrm{I}}}{1-\alpha_{\mathrm{II}}}
\ge
\alpha_{\mathrm{I}}\log \alpha_{\mathrm{I}}.
\]
The function \(x \mapsto x\log x\), attains its minimum over \(x>0\) at \(x=e^{-1},\) with minimum value \(-\frac{1}{e}.\)
Hence $\alpha_{\mathrm{I}}\log \alpha_{\mathrm{I}} \ge -\frac{1}{e}$. Also, since $1-\alpha_{\mathrm{I}}\ge 1-\eta$
and
$\alpha_{\mathrm{II}}\le \eta$, we obtain
\[
(1-\alpha_{\mathrm{I}})
\log\!\frac{1-\alpha_{\mathrm{I}}}{\alpha_{\mathrm{II}}}
\ge
(1-\eta)\log\!\frac{1-\eta}{\eta}.
\]
Therefore
\[
D_{\KL}(P\|Q)
\ge
(1-\eta)\log\!\frac{1-\eta}{\eta}
-\frac{1}{e}.
\]
If \(\eta\le \frac14\), then
$1-\eta \ge \frac34$
and
$\log(1-\eta)\ge \log\frac34$.
Therefore
\begin{align*}
(1-\eta)\log\!\frac{1-\eta}{\eta}
&=
(1-\eta)\log\!\frac{1}{\eta}
\\
&\qquad
+
(1-\eta)\log(1-\eta)
\\
&\ge
\frac34 \log\!\frac{1}{\eta}
+
\frac34 \log\!\frac34.
\end{align*}
Hence
\[
D_{\KL}(P\|Q)
\ge
\frac34 \log\!\frac{1}{\eta}
+
\frac34 \log\!\frac34
-\frac{1}{e}.
\]
Now let \(L:=\log\!\frac{1}{\eta}.\) Since $\eta\le \frac14$, we have
$L\ge \log 4 > \frac54$.
Also,
$\log\!\frac34 > -\frac13$
and $\frac1e < \frac38$. Therefore $\frac34 \log\!\frac34 - \frac1e
>
-\frac14-\frac38
=
-\frac58$. Hence $D_{\KL}(P\|Q)
\ge
\frac34 L-\frac58$. Since \(L>\frac54,\) we have
$\frac12L>\frac58$. Therefore
$\frac34L-\frac58
=
\frac14L+\left(\frac12L-\frac58\right)
\ge
\frac14L$. Therefore
\[
D_{\KL}(P\|Q)
\ge
\frac14 \log\!\frac{1}{\eta},
\]
which proves \eqref{eq:testing_KL_log_eta}.
\end{proof}

Having established the above results, we are now ready to complete the proof of Theorem~\ref{thm:nonad_lb_full}.

\begin{proof}[Proof of Theorem~\ref{thm:nonad_lb_full}]
Fix any non-adaptive protocol and let $\mu$ be its induced allocation over $\cB$. By Lemma~\ref{lem:rare_cylinder}, there exists a rare prefix $u^\star$ satisfying \eqref{eq:rare_cylinder}. Construct the hard pair $\rho^{(0)},\rho^{(1)}$ as in Lemma~\ref{lem:hard_pair_prefix}. By \eqref{eq:two_point_trdist},
$\frac12\|\rho^{(0)}-\rho^{(1)}\|_1=\frac{|\alpha|}{\sqrt{2}}$.
Choose the universal constant $c_0>0$ in the theorem statement small enough that
$2c_0|\alpha| < \frac{|\alpha|}{\sqrt{2}}$. Any $c_0<1/(2\sqrt{2})$ suffices.

Assume now that the protocol satisfies the guarantee in Theorem~\ref{thm:nonad_lb_full}; namely, for every state $\rho$ in the family,
\[
\mathbb{P}_{\rho}\!\left(
\frac12\|\widehat\rho-\rho\|_1 \le c_0|\alpha|
\right)\ge 1-\eta.
\]
In particular, this holds when the true state is $\rho^{(0)}$ and also when the true state is $\rho^{(1)}$.

Since
$\frac12\|\rho^{(0)}-\rho^{(1)}\|_1=\frac{|\alpha|}{\sqrt{2}}$ and
$2c_0|\alpha| < \frac{|\alpha|}{\sqrt{2}}$, the two trace-distance balls of radius $c_0|\alpha|$ around $\rho^{(0)}$ and $\rho^{(1)}$ are disjoint. Therefore, given the estimator output $\widehat\rho$, we may define a binary test by comparing its trace distance to $\rho^{(0)}$ and $\rho^{(1)}$, and choosing whichever of the two is closer. Whenever $\frac12\|\widehat\rho-\rho^{(j)}\|_1 \le c_0|\alpha|$, this rule correctly identifies the true state $\rho^{(j)}$. Hence the resulting test has type-I and type-II errors both at most $\eta$.

Let $P:=P^{(0)}_{\mu,M}$ and
$Q:=P^{(1)}_{\mu,M}$ denote the distributions of the full transcript $T=(B_1,O_1,\dots,B_M,O_M)$ when the true state is $\rho^{(0)}$ and $\rho^{(1)}$, respectively. We now combine the lower and upper bounds on \(D_{\KL}(P\|Q)\). Lemma~\ref{lem:testing_KL_lower} gives the lower bound required by success probability at least \(1-\eta\), while Lemma~\ref{lem:kl_contract_nonad} gives the corresponding upper bound in terms of \(M\), \(\mu(\cB_{u^\star})\), and \(\mathrm{kl}(\alpha)\). Lemma~\ref{lem:testing_KL_lower} gives
\[
D_{\KL}(P\|Q)\gtrsim \log\!\Big(\frac{1}{\eta}\Big).
\]
For each of the two settings \(b\in\{b^{(0)},b^{(1)}\}\), the probabilities under \(\rho_{\mathrm{pref}}\) take the form
\[
\begin{aligned}
p(o\mid b,\rho_{\mathrm{pref}})
&=
2^{-n}
\left(
1+\sum_{k=1}^{n-1}\beta_k s_k(o)
\right),
\\
&\qquad
s_k(o)\in\{\pm 1\}.
\end{aligned}
\]
Therefore
\[
p(o\mid b,\rho_{\mathrm{pref}})
\ge
2^{-n}
\left(
1-\sum_{k=1}^{n-1}|\beta_k|
\right)
\ge
\delta\,2^{-n}.
\]
Thus the hypothesis \eqref{eq:baseline_lower_bound} in Lemma~\ref{lem:kl_contract_nonad} holds, and hence
\[
D_{\KL}(P\|Q)
\le
\frac{C\,M}{3^{n-1}}\mathrm{kl}(\alpha).
\]
Combining this with the lower bound from Lemma~\ref{lem:testing_KL_lower} gives
\[
M
\gtrsim
\frac{3^{n-1}}{\mathrm{kl}(\alpha)}
\log\!\Big(\frac{1}{\eta}\Big).
\]
If $|\alpha|$ is sufficiently small, then \(\mathrm{kl}(\alpha)=\Theta(\alpha^2)\), so this is equivalently
\[
M
=
\Omega\!\Big(\frac{3^n}{\alpha^2}\log(1/\eta)\Big),
\]
up to absolute constants. This proves the theorem.
\end{proof}

\begin{proof}[Proof of Corollary~\ref{cor:adaptive_nonadaptive_separation} statement $(ii)$]

Under the choice
\[
\alpha=\frac{\varepsilon}{4},
\qquad
\beta_k=\frac{\varepsilon}{4(n-1)},
\qquad
k=1,\dots,n-1,
\]
Section~\ref{sec:prefix_family} shows that the family is physical. Moreover,
\[
\sum_{k=1}^{n-1}|\beta_k|
=
(n-1)\cdot \frac{\varepsilon}{4(n-1)}
=
\frac{\varepsilon}{4}
\le
\frac14,
\]
since \(\varepsilon\in(0,1)\). Thus the coefficient condition in Theorem~\ref{thm:nonad_lb_full} holds with \(\delta=\frac34\).

We may therefore apply Theorem~\ref{thm:nonad_lb_full}. Since \(\alpha=\frac{\varepsilon}{4},\) and \(\mathrm{kl}(\alpha)=\Theta(\alpha^2)=\Theta(\varepsilon^2)\) for small \(\varepsilon\), it follows that every non-adaptive protocol that is \((c\varepsilon,\eta)\)-accurate over the family must satisfy
\[
M_{\mathrm{nonad}}
\ge
c\,3^{\,n-1}\varepsilon^{-2}\log\!\Big(\frac{1}{\eta}\Big)
=
\Omega\!\big(3^n\varepsilon^{-2}\log(1/\eta)\big),
\]
for a universal constant \(c>0\). This proves the corollary.
\end{proof}

\section{Numerical Illustration of the Scaling Gap}
\label{sec:numerical_illustration}

The previous sections establish explicit upper and lower bounds for adaptive and non-adaptive procedures under the same family and measurement model. To complement these theorem-level guarantees, we now examine how the predicted scaling gap appears numerically at finite problem sizes.

We use the concrete coefficient choice from Corollary~\ref{cor:adaptive_nonadaptive_separation}, with fixed accuracy scale $\varepsilon = 0.5$, and report only the likelihood-based decision rule (LLR). For the adaptive method, we simulate the same stagewise prefix-recovery strategy analyzed in the paper. For the non-adaptive baseline, we use a fully non-adaptive design that allocates measurements uniformly over all $3^n$ product-Pauli bases. For each value of $n$, we sweep the shot parameter $m$, estimate the success probability by Monte Carlo simulation, and record the smallest total budget achieving
\[
P_{\mathrm{success}} \ge 0.90.
\]

Figure~\ref{fig:llr_adaptive_budget} shows the required adaptive budget as a function of $n$. The numerical curve follows a clear polynomial trend and is well approximated by a cubic fit, consistent with the dependence on $n$ predicted by Corollary~\ref{cor:adaptive_nonadaptive_separation}.

\begin{figure}[t]
\centering
\includegraphics[width=0.84\linewidth]{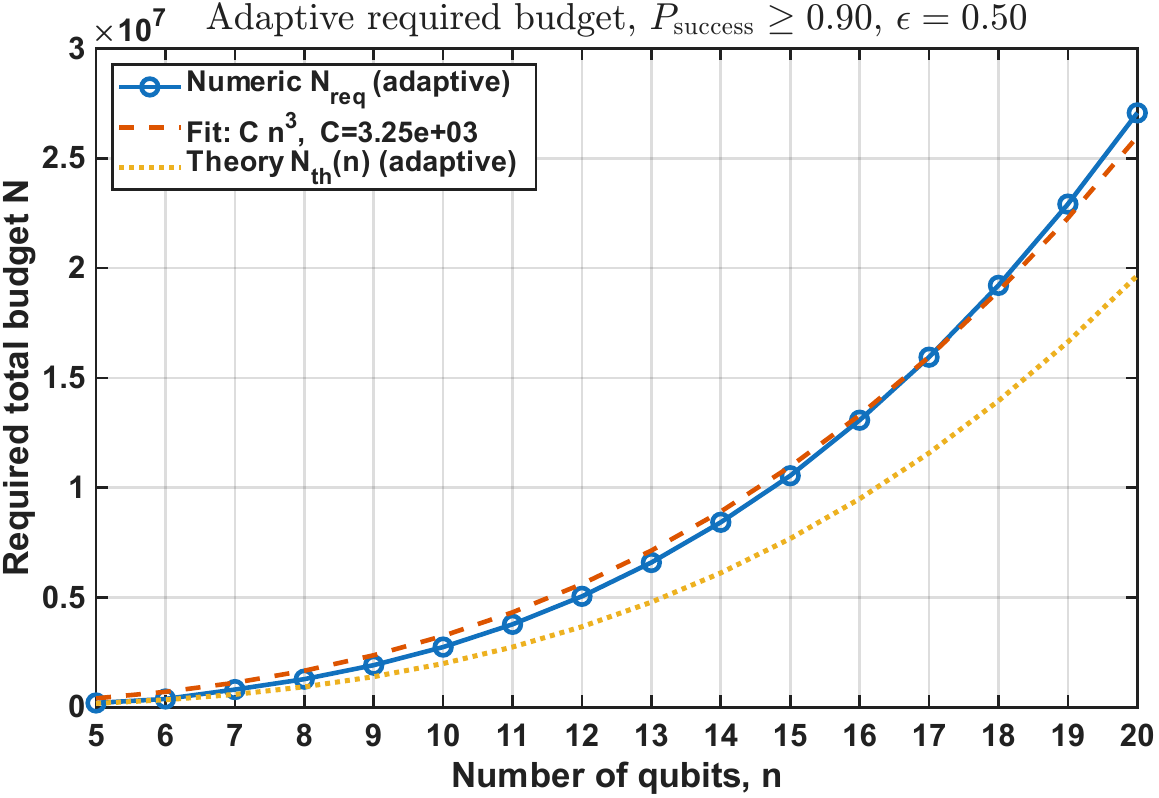}
\caption{Required adaptive budget for the LLR-based prefix-recovery rule, defined as the smallest total budget achieving $P_{\mathrm{success}} \ge 0.90$ at $\varepsilon = 0.5$. The numerical curve is well approximated by a cubic fit and is consistent with the polynomial trend predicted by the adaptive upper bound.}
\label{fig:llr_adaptive_budget}
\end{figure}

Figure~\ref{fig:llr_nonadaptive_budget} shows the corresponding non-adaptive baseline. In contrast to the adaptive case, the required budget grows very rapidly with $n$ and is well described by an exponential fit, in qualitative agreement with the worst-case non-adaptive lower bound.

\begin{figure}[t]
\centering
\includegraphics[width=0.84\linewidth]{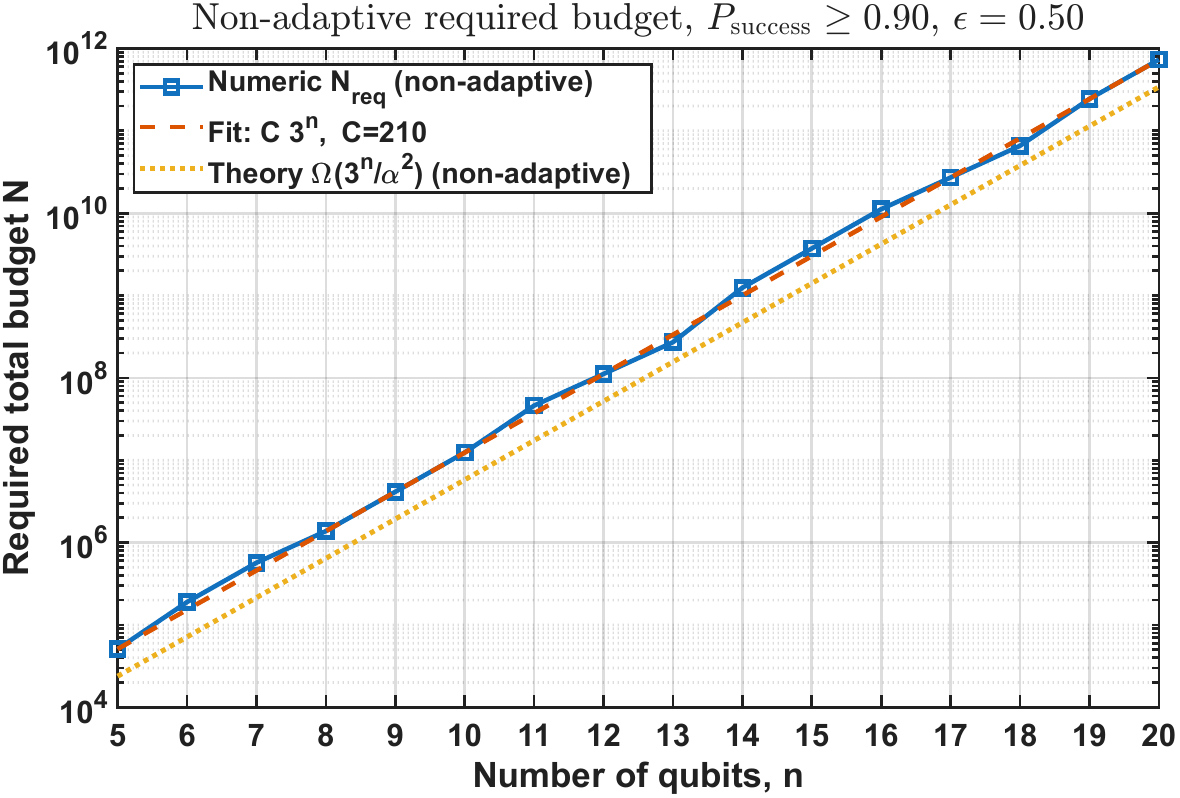}
\caption{Required non-adaptive budget for the LLR-based baseline, again defined by the threshold $P_{\mathrm{success}} \ge 0.90$ at $\varepsilon = 0.5$. The numerical curve exhibits clear exponential growth in $n$, in qualitative agreement with the worst-case non-adaptive lower bound.}
\label{fig:llr_nonadaptive_budget}
\end{figure}

The comparison is shown directly in Figure~\ref{fig:llr_direct_comparison}. As $n$ increases, the exponential growth of the non-adaptive budget rapidly dominates, while the adaptive budget remains on a much milder polynomial trajectory.

\begin{figure}[t]
\centering
\includegraphics[width=0.84\linewidth]{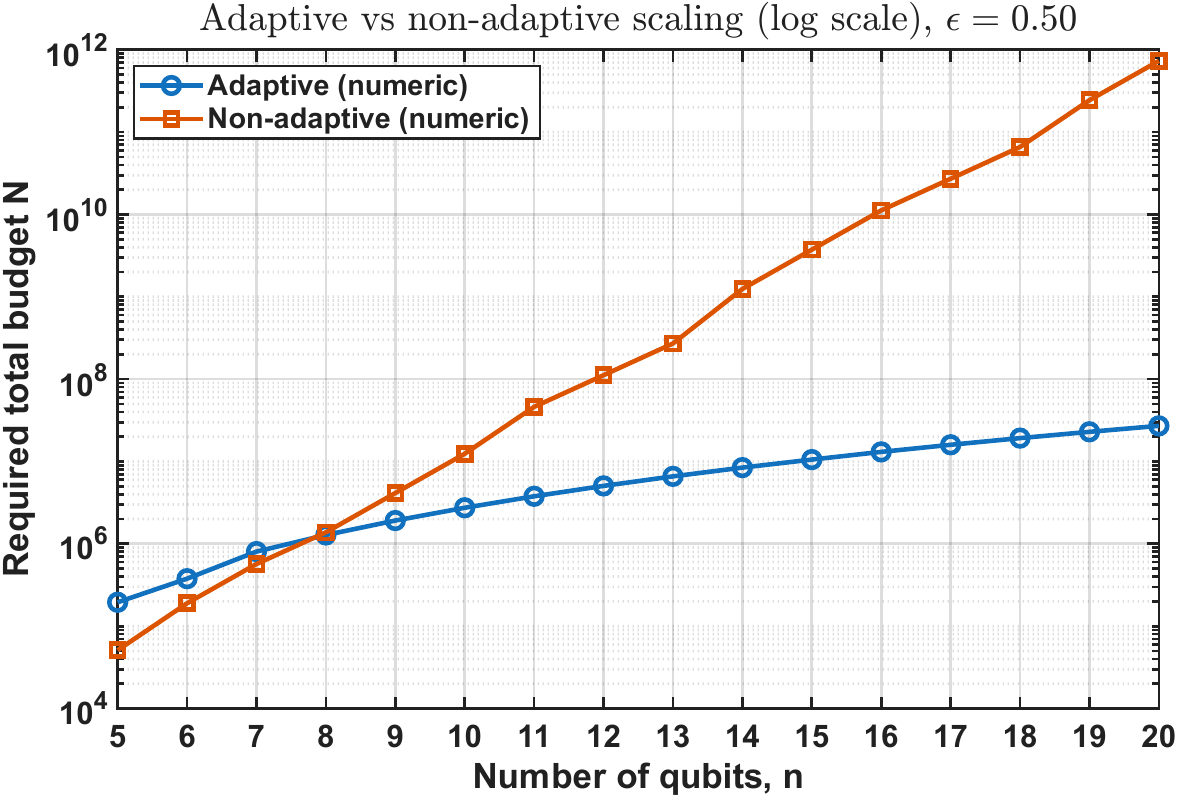}
\caption{Direct comparison of adaptive and non-adaptive required total budgets on the same scale. The adaptive LLR-based procedure remains on a much milder growth trajectory, while the non-adaptive baseline exhibits a rapid exponential increase.}
\label{fig:llr_direct_comparison}
\end{figure}

Taken together, these experiments provide a numerical illustration of the same scaling gap proved analytically in the paper. Under the present construction and measurement model, the adaptive advantage is clearly visible in the numerical results.

\section{Discussion and Future Directions}
\label{sec:discussion}

The results of this paper identify a concrete regime in which adaptivity changes the worst-case sample-complexity scaling. The regime is deliberately specific: the measurement architecture is fixed to single-copy tensor-product Pauli measurements, the target is restricted to the structured prefix/tree family \(\cF_n(\alpha,\beta)\), and success is measured in trace distance with high probability in a minimax sense over that family. Within this setting, the separation is sharp: the adaptive procedure succeeds with polynomial copy complexity, whereas every non-adaptive design requires exponentially many copies in the worst case.

This conclusion also helps situate the result within the broader tomography literature. The effect of adaptivity depends on the state family, the measurement model, and the performance criterion, and different choices of these ingredients can lead to different conclusions. In the highly structured regime studied here, adaptivity changes the scaling in a fundamental way. In this sense, the present result identifies a concrete tomography setting in which sequential basis selection yields a genuine worst-case advantage.

More broadly, this result fits into the growing effort to develop scalable alternatives to full tomography. This includes recent work on quantum-state learning complexity, methods for testing and certifying states, and classical-shadow-style prediction methods~\cite{AnshuArunachalam2024Survey,EisertHangleiterWalkRothMarkhamParekhChabaudKashefi2020,HuangKuengPreskill2020}.

We introduced the prefix/tree family as a structured construction to make the origin of the separation transparent. In this setting, an adaptive protocol can exploit hierarchical intermediate information stage by stage, whereas non-adaptive designs remain constrained by the rare-prefix bottleneck in the worst case. The construction therefore illustrates a broader theme in structured tomography: local measurements can become substantially more informative when the underlying state family exposes usable intermediate structure.

Several nearby questions remain open. One direction is to study whether similar gaps appear under average-case performance criteria, rather than worst-case minimax guarantees. Another is to understand whether the same basic separation mechanism can arise in structured families beyond the prefix/tree construction studied here.

\begin{acknowledgements}
We acknowledge funding support from NSF Grants No. CCF-2106834, CCF-2241298, and ECCS-2409701.
\end{acknowledgements}
\bibliographystyle{quantum}
\bibliography{mybibliography}

\end{document}